\documentclass{article}
\usepackage{amssymb, amsfonts, amsmath, amsthm}
\usepackage{euscript}
\usepackage{color}
\usepackage{algorithm, algorithmic}
\usepackage[square, comma, sort&compress, numbers]{natbib}

\def\theenumi{\roman{enumi}}
\def\labelenumi{\bf (\theenumi)}

\theoremstyle{plain}
\newtheorem{thm}{Theorem}[section]
\newtheorem{wthm}[thm]{Incorrect Theorem}
\newtheorem{cor}[thm]{Corollary}
\newtheorem{lem}[thm]{Lemma}
\newtheorem{prop}[thm]{Proposition}
\theoremstyle{definition}
\newtheorem{ex}[thm]{Example}

\newtheorem{nota}[thm]{Notation}

\newcommand{\F}{\mathbb{F}}
\newcommand{\z}{\mathbb{Z}}
\newcommand{\se}{\subseteq}

\newcommand{\dis}{\oplus}

\newcommand{\sm}{\setminus }
\newcommand{\ifof}{if and only if }
\renewcommand{\iff}{\Leftrightarrow}
\newcommand{\give}{$\Rightarrow$}
\newcommand{\rgive}{$\Leftarrow$}
\renewcommand{\l}{\left}
\renewcommand{\r}{\right}
\newcommand{\rg}{\rangle}
\renewcommand{\lg}{\langle}
\newcommand{\p}[1]{\left(#1\right)}
\newcommand{\ceil}[1]{\left\lceil #1\right\rceil}
\newcommand{\f}[2]{\frac{#1}{#2}}

\newcommand{\li}{L_{\mathrm{ind}}}
\newcommand{\lc}{L_{\mathrm{coef}}}

\renewcommand{\mod}[1]{\ \l(\mathrm{mod}\ #1\r)}
\newcommand{\mat}[2]{\l(\begin{array}{#1} #2 \end{array}\r)}
\newcommand{\fm}{\mathfrak{M}}

 \notag
\begin{document}

\title{Dual of Codes over Finite  Quotients of Polynomial Rings}
\author{Ashkan Nikseresht \\
\it\small Department of Mathematics, Institute for Advanced Studies in Basic Sciences,\\
\it\small P.O. Box 45195-1159, Zanjan, Iran\\
\small E-mail:  ashkan\_nikseresht@yahoo.com%
}
\date{}

\maketitle

\let\oldref=\ref
\def\ref#1{(\oldref{#1})}

\begin{abstract}
Let $A=\f{\F[x]}{\lg f(x)\rg }$, where $f(x)$ is a monic polynomial over a finite field $\F$.  In this paper, we
study the relation between $A$-codes and their duals. In particular, we state a counterexample and a correction
to a theorem of Berger and  El Amrani (Codes over finite quotients of polynomial rings, \emph{Finite Fields
Appl.} \textbf{25} (2014), 165--181) and present an efficient algorithm to find a system of generators for the
dual of a given $A$-code. Also we characterize self-dual $A$-codes of length 2 and investigate when the
$\F$-dual of $A$-codes are $A$-codes.
\end{abstract}


                                         \section{Introduction}

Throughout this paper $A=\f{\F[x]}{\lg f(x)\rg }$, where $f(x)$ is a monic polynomial over a finite field $\F$.
Moreover, $\deg(f)=m$ and $|\F|=q$. We consider elements of $A$ as polynomials of degree $<m$ where the arithmetic
is done modulo $f(x)$. By a \emph{linear $A$-code} (an \emph{$A$-code}, for short) of length $l$ we mean an
$A$-submodule of $A^l$.

In the case $f(x)=x^m-1$ and $l=1$, $A$-codes are the well-known cyclic $q$-ary codes. Also if $l>1$ with $f(x)=
x^m-1$, then $A$-codes represent quasi-cyclic codes over $\F$ which have recently gained great attention (see, for
example \cite{QC bound parity, QC gen C, QC impel, QC lin alg, qc1, qc2}). Also in the case that $f(x)$ is a power
of an irreducible polynomial, then $A$ is a finite chain ring and codes over such rings have attracted a lot of
researchers (see for example \cite{FCR consta,FCR LCD,FCR self-dual}).

In \cite{qc1}, a canonical generator matrix for quasi-cyclic codes is given, when these codes are viewed as
$A$-codes with $f(x)=x^m-1$. In \cite{main} these results are generalized to arbitrary $A$-codes. Let
$C^{\bot}=\{(a_1, \ldots, a_l)\in A^l|\forall c\in C\quad \sum_{i=1}^l a_ic_i=0\}$ be the dual of an $A$-code $C$.
Section 2.6 of \cite{main} states how to compute a system of generators of $C^{\bot}$. In Section 2, we will show
that the main theorem of \cite[Section 2.6]{main} is not correct and we state a correction of this theorem. Also
we present an efficient algorithm to find a \emph{generator matrix} for $C^\bot$ (that is, a matrix, rows of which
generate $C^\bot$ as an $A$-module).

In Section 3, we apply our results to find all self-dual $A$-codes with length $\leq 2$ and self-dual $A$-codes
which have a basis of divisors containing just one element.

Every $A$-code $C$ of length $l$ could be seen as an $\F$-code of length $ml$ (by replacing $a(x)\in A$ with the
sequence of its coefficients). Therefore we can form the $\F$-dual of $C$. The $\F$-dual of an $A$-code is not
always an $A$-code (see \cite[Example 7]{main}). In Section 4, we characterize rings $A$, such that the $\F$-dual
of every $A$-code is an $A$-code and also rings over which the $\F$-dual and the $A$-dual of codes coincide.

Before stating the main assertions, let's recall some notations and results form \cite{main}, which will be used
later.
\paragraph{A brief review of bases of divisors of an $A$-code.}

Assume that $C\neq \bf{0}$ is an $A$-code of length $l$ and $u=(u_1(x), u_2(x), \ldots, u_l(x))\in A^l$. The
\emph{leading index} of $u$, denoted $\li(u)$ is the smallest integer $i$ such that $u_i\neq 0$ and
$\lc(u)=u_{\li(u)}$ is called the \emph{leading coefficient} of $u$ (we set $\li(0)=\infty$). Also by $\li(C)$ we
mean $\min\{\li(u)|u\in C\}$ and $\lc(C)$ is the single monic polynomial $g(x)$ with the minimum degree such that
there is a $c\in C$ with $\li(c)=\li(C)$ and $\lc(c)=g(x)$. An element $c\in C$ satisfying this condition is
called a \emph{leading element} of $C$.

Recursively set $C^{(1)}=C$ and if $\li(C^{(n)})\leq l$, then $$C^{(n+1)}=\{c\in C^{(n)}| \li(c)>
\li(C^{(n)})\}.$$ Let $k$ be largest integer such that $C^{(k)}\neq \{0\}$ and assume that for $1\leq j \leq k$,
$g^{(j)}$ is a leading element of $C^{(j)}$. Then by Theorem 1 and Proposition 2 of \cite{main}, $C$ is generated
by $B=(g^{(1)}, \ldots, g^{(k)})$ (as an $A$-module) and $k$ and $\deg (\lc(g^{(i)})$s are independent of the
choice of $g^{(i)}$s. Also $|C|=q^\alpha$ where $\alpha=km- \sum_{i=1}^k \deg (\lc(g^{(i)}))$. Any $B$ as above is
called a \emph{basis of divisors} of $C$.

Now let $G$ be the matrix whose $i$-th row is $g^{(i)}$. Suppose that $g_{i,j_i}$ is the leading coefficient of
the $i$-th row of $G$. If $G$ has the property that $\deg(g_{t,j_i})< \deg(g_{i,j_i})$ for all $1\leq i\leq k$ and
$t<i$, then $G$ is called the \emph{canonical generator matrix} (\emph{CGM}, for short) of $C$ and $B$ is called
the \emph{canonical basis of divisors} of $C$. In \cite[Theorem 2]{main} it is shown that every $A$-code has a
unique CGM. Also they present algorithms to find a basis of divisors and the CGM of a a given $A$-code.

\begin{ex}
Suppose that $f(x)=x(x^2+x+1)$ and $C$ is the submodule of $A^3$ generated by $g=(x^2, 0, x^2+1)$. Then $\li(C)=
\li(g) =1$. By definition $C^{(1)}=C$. To compute $C^{(2)}$ we should find all elements of $C$ whose leading index
is greater that $\li(C^{(1)})=1$, that is, all elements of $C$ with zero on the first component. Since every
element $c$ of $C$ is of the form $c=a(x) g$ for some $a(x)\in A$, we see that $c_1=0 \iff x^2+x+1\mid a(x) \iff
c=a'(x) ((x^2+x+1) g)$ for some $a'(x)\in A$ which is equivalent to $c=a'(x)(0, 0, x^2+x+1)$. Therefore, $C^{(2)}$
is the $A$-code generated by $g^{(2)}=(0, 0, x^2+x+1)$ and $\li(C^{(2)})=3$. Also the only element in $C^{(2)}$
whose leading index is $>3$, is zero, hence $C^{(3)}=0$ and $k=2$.

It is clear that the (only) leading element of $C^{(2)}$ is $g^{(2)}$. But as $x^2\nmid f(x)$, $g$ is not a
leading element of $C^{(1)}$. Indeed, since gcd$(x^2, f(x))=x$ and $(-x-1)g=(x, 0, -1)$, we deduce that a leading
element of $C^{(1)}$ is $g^{(1)}= (x, 0, -1)$ (note that this element is not unique, for example $g^{(1)}+g^{(2)}$
is another leading element of $C^{(1)}$). Therefore $(g^{(1)}, g^{(2)})$ is a basis of divisors of $C$ and it
follows that $\mat{ccc}{x& 0& -1 \cr 0 & 0 & x^2+x+1}$ is the CGM of $C$. Also $\dim_\F(C)=
6-\deg(x)-\deg(x^2+x+1)=3$.
\end{ex}

                    \section{A Generator Matrix for the Dual of an $A$-Code}

We start by  presenting a counterexample of \cite[Theorem 3]{main} and stating a correction of this theorem. Then
we use this correction to give an algorithm which generates a generator matrix for the dual of an $A$-code.

Throughout this section, without any further mention, we assume that $C$ is an $A$-code of length $l$ and that
$g^{(1)}=(g_{1,1}(x), \ldots, g_{1,l}(x))$ is the first element of its canonical basis of divisors. Also we let
$C'$ be the punctured code of $C^{(2)}$ in the first position and assume that $G'$ is the canonical generator
matrix of $C'$. Note that $G'$ is the matrix obtained by deleting the first row and column of the canonical
generator of $C$. The following theorem is claimed to be proved in \cite{main}.

\begin{wthm}[{\cite[Theorem 3]{main}}] \label{wrong thm}
Suppose that $\li(C)=1$ and $h_{1,1}(x)= \f{f(x)}{g_{1,1}(x)} \mod{f(x)}$. Let $H'$ be a generator matrix of
$C'^{\bot}$. Then
$$H=\l(\begin{array}{c|c}
1 & 0 \ldots 0 \\
\hline
\begin{array}{c} 0\\ \vdots \\ 0 \end{array} & H'
\end{array}\r) \times
\l(\begin{array}{ccccc}
h_{1,1} &0 & \ldots & \ldots & 0 \\
-g_{1,2} & g_{1,1} & \ddots & \ddots & \vdots \\
-g_{1,3} & 0 & g_{1,1} & \ddots & \vdots \\
\vdots & \vdots & \ddots & \ddots & 0 \\
-g_{1,l} & 0 & \ldots & 0 & g_{1,1}
\end{array}\r) $$
is a generator matrix for $C^{\bot}$.
\end{wthm}

To present a counterexample of \ref{wrong thm}, we need the following result. We say an element of $A^l$ is monic
when its leading coefficient is monic.
\begin{prop}\label{div bas char}
Let $G$ be a $k\times l$ generator matrix for an $A$-code $C$. Suppose that $g^{(i)}=$ the $i$-th row of $G$, is
monic. Then $(g^{(1)},\ldots, g^{(k)})$ is a basis of divisors of $C$ \ifof the following hold.
\begin{enumerate}
\item \label{div bas char 1} $G$ is in echelon form.

\item \label{div bas char 2} $\lc\p{g^{(i)}}|f(x)$.

\item \label{div bas char 3} $h_i g^{(i)}$ is an $A$-linear combination of $g^{(i+1)}, \ldots , g^{(k)}$ where
    $h_i(x)= \f{f(x)}{\lc\p{g^{(i)}}}$.
\end{enumerate}
Moreover, if we replace \ref{div bas char 3} with \ref{div bas char 3'} below, the assertion remains valid.
\def\theenumi{\roman{enumi}$'$}
\def\labelenumi{\bf (\theenumi)}

\begin{enumerate}
\addtocounter{enumi}{2}
\item \label{div bas char 3'} $\dim_\F C= \sum_{i=1}^k m-\deg(\lc(g^{(i)}))$.
\end{enumerate}
\def\theenumi{\roman{enumi}}
\def\labelenumi{\bf (\theenumi)}
\end{prop}
\begin{proof}
(\give): \ref{div bas char 1} follows from the definition of $C^{(i)}$ and $g^{(i)}$. \ref{div bas char 2} follows
from the remarks above Definition 5 of \cite[p. 170]{main}. Let $l_i=\li(g^{(i)})$, then by the definition of
$h_i$, the $l_i$-th entry of $h_ig^{(i)}$ is 0 in $A$, hence $h_ig^{(i)}\in C^{(i+1)}$ and \ref{div bas char 3}
follows.

(\rgive): First we prove that for each $i$, $C^{(i)}$ is generated by the set $B=\{g^{(i)}, g^{(i+1)}, \ldots,
g^{(k)}\}$. We prove this for $i=2$ and the rest follows by induction. Since $G$ is in echelon form, $\lg B\rg \se
C^{(2)}$. Let $g$ be an arbitrary element of $C^{(2)}$. Then $g=\sum_{i=1}^k a_i(x) g^{(i)}$. Suppose that
$a'_1(x)= a_1(x) \mod{h_1(x)}$. Then by \ref{div bas char 2}, $(a_1(x)-a'_1(x)) g^{(1)}\in \lg B\rg$, hence
$$a'_1(x)g^{(1)}+ \lg B\rg= g+\lg B \rg \se C^{(2)}.$$ As $g_{1,l_1}$ is monic ($l_i$s as in (\give)) and $\deg
(a'_1(x))< \deg (h_1(x))$, if $a'_1(x) \neq 0$, then $a'_1(x)g_{1,l_1}(x)\neq 0$ and $\li(a'_1(x)g^{(1)})=l_1$
contradicting $a'_1(x)g^{(1)}\in C^{(2)}$. Therefore $a'_1(x)=0$ and $g\in \lg B\rg$ as required.

Now it is clear that $l_i=\li(C^{(i)})$ for each $i$ and if $g\in C$ with $\li(g)=l_i$, then $\lc(g)=a(x)
\lc(g^{(i)})$ for some $a(x)\in A$ with $\deg a(x)<\deg h_i(x)$. Therefore $g^{(i)}$ is a leading element of
$C^{(i)}$.

For the ``moreover'' statement, note that if  $(g^{(1)},\ldots, g^{(k)})$ is a basis of divisors, then by
\cite[Proposition 2]{main}, \ref{div bas char 3'} holds. Conversely if \ref{div bas char 1} and \ref{div bas char
2}  hold, then clearly the combinations of the form $\sum_{j=1}^k z_jg^{(j)}$ for $z_j\in A$ with $\deg(z_j)<\deg
(h_j)= m-\deg(\lc(g^{(j)}))$ are mutually different elements of $C$. So if \ref{div bas char 3'} also is valid,
then these combinations are all elements of $C$. In particular, $c= h_ig^{(i)}$ could be written as such a
combination and since  for each $j \leq \li(g^{(i)})$ the $j$-th entry of $c$ is zero, we get $z_j=0$ for $j\leq
i$, as required.
\end{proof}

\begin{ex}\label{counterexam}
Let $\F=\F_2, f(x)=x^2(x^3+1)$ and $C$ be the $A$-code of length 3 which is generated by
$$G=\l(\begin{array}{ccc}
 x & x & 0 \\
 0 & x^2 & 1 \\
 0 & 0 & x^3+1
\end{array}\r).$$
Using \ref{div bas char}, we can see that $G$ is the CGM of $C$. As $(x,x, 0)\in C$, we have $\li(C)=1$ and the
assumptions of Theorem \ref{wrong thm} are valid. Also $C'$ is generated by $(x^2,1)$ and $(0,x^3+1)$. One can
readily check that a generator matrix for $C'^\bot$ is $H'=(1\ x^2)$. Thus if $H$ is as in \ref{wrong thm}, then
$$H=\mat{ccc}{
 1 & 0& 0 \cr
 0 & 1 & x^2
}\times \mat{ccc}{
 x(x^3+1) & 0 & 0 \cr
 x & x & 0\cr
 0 & 0 & x
}=\l(\begin{array}{ccc}
x(x^3+1) & 0 & 0 \\
x & x & x^3
\end{array}\r).$$
Clearly $u=(1,1,x^2)\in C^\bot$. But if $u$ is a linear combination of the rows of $H$, then for some $a(x),
b(x)\in \F[x]$ we have $a(x) x(x^3+1) +b(x) x \equiv 1 \mod{f(x)}$ which leads to $x|1$, a contradiction. Thus $H$
is not a generator matrix of $C^\bot$ and \ref{wrong thm} is not correct.
\end{ex}
To present the correct generator matrix for $C^{\bot}$ we need the following lemma.
\begin{lem}\label{lem1}
The code $C'^\bot$ is the punctured code of $C^\bot$ in the first position.
\end{lem}
\begin{proof}
We must show that for each $c'\in C'^\bot$ there is a $c_1\in A$ such that $(c_1|c')$ (the concatenation of $c_1$
to $c'$) is an element of $C^\bot$. If $\li(C)>1$, then any $c_1\in A$ works. Assume that $\li(C)=1$ and let
$c=(a_1, \ldots, a_l)\in A^l$. Define $\phi(c)=(a_2, \ldots, a_l)$. Then $\phi:C^\bot \to C'^\bot$ is a
$\F$-linear map and it suffices to show that $\phi$ is onto.

Suppose that $(g^{(1)}, \ldots, g^{(k)})$ is the canonical basis of divisors of $C$ and set $r_i=\deg(
\lc(g^{(i)}))$. Note that $\ker \phi=\{(c_1, 0, 0, \ldots, 0)\in A^l| c_1 g_{1,1}=0\}$. If $h_{1,1}= f/g_{1,1}$,
then $c_1g_{1,1}=0 \iff c_1=c'_1h_{1,1}$ for some $c'_1\in \F[x]$ with $\deg (c'_1)<\deg (g_{1,1})$. Thus $\dim_\F
\ker \phi=\deg (g_{1,1})=r_1$. According to \cite[Proposition 2]{main}, $\dim_\F C^\bot= lm-\dim_\F C=
lm-\sum_{i=1}^k (m-r_i)$. Similarly $\dim_\F C'^\bot=(l-1)m-\sum_{i=2}^k (m-r_i)$. Hence $$\dim_\F \phi(C^\bot)=
\dim_\F C^\bot-\dim_\F \ker \phi=\dim_\F C'^\bot$$ and hence $\phi$ is onto.
\end{proof}

\begin{thm}\label{main thm}
Assume that $\li(C)=1$, $l>1$ and $h_{1,1}(x)= \f{f(x)}{g_{1,1}(x)} \mod{f(x)}$. Let $H'=(h'_{ij})_{2\leq i\leq
k}^{2\leq j\leq l}$ be a generator matrix of $C'^{\bot}$. A generator matrix for $C^\bot$ is
$$H=\l(\begin{array}{c|c}
h_{1,1} & 0 \ldots 0 \\
\hline
  \begin{array}{c}
  \alpha_2 \\
  \alpha_3 \\
  \vdots \\
  \alpha_k
  \end{array}
  & H'
\end{array}\r),$$
where $$\alpha_i= -\f{\sum_{j=2}^l h'_{ij}g_{1j}}{g_{1,1}} \mod{h_{1,1}}. $$
\end{thm}
\begin{proof}
First note that by the previous lemma, for each $2\leq i\leq k$, there is an $a_i\in A$, such that $(a_i, h'_{i2},
h'_{i3}, \ldots, h'_{i,l}) \in C^\bot$. This means that $a_ig_{1,1}+ \sum_{j=2}^l h'_{ij}g_{1j} =0$ in $A$. Thus
in $\F[x]$ we have $g_{1,1}| \sum_{j=2}^l h'_{ij}g_{1j} +bf$ for some $b\in \F[x]$. Therefore $g_{1,1}|
\sum_{j=2}^l h'_{ij}g_{1j}$ in $\F[x]$ for each $2\leq i\leq k$ and $\alpha_i$s are well defined.

Denote the $i$-th row of $H$ and $H'$ by $h^{(i)}$ and $h'^{(i)}$, respectively. It is easy to see that $h^{(i)}$s
are in $C^\bot$. Conversely, let $c=(c_1, \ldots, c_l)\in C^\bot$, then $c'=(c_2, \ldots, c_l)\in C'^\bot$. Thus
for some $\lambda_2, \ldots, \lambda_k\in A$, we have $c'=\sum_{i=2}^k \lambda_i h'^{(i)}$. Note that in $A$,
$g_{1,1} \alpha_i= -\sum_{j=2}^l h'_{ij}g_{1j}$. So
\begin{align}
 g_{1,1}\sum_{i=2}^k \lambda_i\alpha_i & =- \sum_{i=2}^k \lambda_i \sum_{j=2}^l  h'_{ij}g_{1j}
 =- \sum_{j=2}^l \p{ \sum_{i=2}^k \lambda_i   h'_{ij} } g_{1j}= - \sum_{j=2}^l c_j g_{1j}=c_1g_{1,1},
\end{align}
where the last equality follows from $c\cdot g^{(1)}=0$. We conclude that $g_{1,1} \p{ c_1- \sum_{i=2}^k
\lambda_i\alpha_i } =0$, that is, $h_{1,1}|c_1- \sum_{i=2}^k \lambda_i\alpha_i$, say $\lambda_1 h_{1,1} = c_1-
\sum_{i=2}^k \lambda_i\alpha_i$. Consequently, $c= \sum_{i=1}^k \lambda_i h^{(i)}$ and $H$ is a generator matrix
for $C^\bot$.
\end{proof}

It should be noted that the above theorem is correct when $C'=0$, in which case $H'=I_{l-1\times l-1}$. Also if
$l=1$, then clearly $H=(h_{1,1})$ is the generator matrix of $C^\bot$. If $\li(C)>1$ and $C_1$ is obtained by
puncturing $C$ in the first position, then it can be seen that $C^\bot=A\dis C_1^\bot$ and hence $H=\mat{c|c}{
 1 & 0 \cr
 \hline
 0 & H'}$ is a generator matrix for $C^\bot$, where $H'$ is a generator matrix for $C_1^\bot$.
Another fact about the previous theorem that should be mentioned is that if we compute $\alpha_i$s modulo $f(x)$
instead of $h_{1,1}(x)$, by the same proof the statement still remains true. The difference is that in the current
form we have $\deg \alpha_i< \deg h_{1,1}$, which will be used in \ref{rev canon}.

Using \ref{main thm} we get the following recursive algorithm for computing a generator matrix of $C^\bot$. In
each recursion of this algorithm the length of the input code is reduced by one. Also to construct the matrix $H$
from $H'$ (in line 19), since $H'$ has dimensions $\leq (l-1)\times(l-1)$, at most $O(l^2)$ polynomial
multiplications in $A$ are performed. Clearly other computations are also bounded above by this bound. Hence
totally, the running time of this algorithm is bounded above by $O(l^3)$ multiplications in $A$. Noting that the
Gaussian elimination method for solving linear equations (even over a field) has the same time complexity
($O(l^3)$ arithmetic operations in the coefficient ring), we see that Algorithm \oldref{alg1} is in fact an
efficient algorithm.
\begin{algorithm}[htb]
\caption{gen-mat-dual$(G)$ (Calculates a generator matrix of dual of an $A$-code $C$)} \label{alg1}
\begin{algorithmic}[1]
 \REQUIRE A generator matrix $G_{k\times l}$ of $C$, rows of which form a basis of divisors for $C$
 \ENSURE A generator matrix $H$ of $C^\bot$

 \IF{the first column of $G$ is zero}
 \IF{l=1}
 \STATE return $H=(1)$
 \ELSE
 \STATE set $G'$ to be $G$ with the first column deleted
 \STATE $H'=$gen-mat-dual$(G')$
 \STATE return $H=\l(\begin{array}{c|c}
                    1 & 0 \\
                    \hline
                    0 & H'
                    \end{array}\r)$
 \ENDIF
 \ELSE \IF{$l=1$}
 \STATE return $H=\p{\f{f(x)}{g_{1,1}(x)} \mod{f(x)}}$
 \ELSE
 \IF{$k=1$ (that is, $C'=0$)}
 \STATE $H'=I_{(l-1)\times (l-1)}$
 \ELSE
 \STATE let $G'$ be $G$ with the first row and column deleted
 \STATE $H'=$gen-mat-dual$(G')$
 \ENDIF
 \STATE construct and return $H$ as in \ref{main thm}
 \ENDIF
 \ENDIF
\end{algorithmic}
\end{algorithm}

The matrix generated by this algorithm is not in the canonical form. But since we calculated the $\alpha_i$s
modulo $h_{1,1}(x)$ instead of $f(x)$ in \ref{main thm}, this matrix is very similar to the canonical form --- we
should just delete the zero rows and then look at the rows and columns in the reverse order. More concretely, we
have the following theorem.
\begin{thm}\label{rev canon}
Let $H_{k\times l}$ be the generator matrix of $C^\bot$ calculated by Algorithm \oldref{alg1} after deleting the
possible zero rows. Let $H^R$ be the $k\times l$ matrix with $h^R_{ij}=h_{k-i+1,l-j+1}$. Then $H^R$ is the CGM of
$C^{\bot R}$, the reciprocal dual of $C$ (that is, $\{(c_l,\ldots,c_1)| (c_1,\ldots,c_l)\in C^\bot\})$.
\end{thm}
\begin{proof}
It follows easily by induction that if $h_{i,j_i}$ is the last nonzero entry on the $i$-th row of $H$, then $j_1<
j_2< \cdots <j_k$ and $h_{i, j_i}|f$ and  $\deg( h_{ij}) <\deg (h_{i, j_i})$ for each $1\leq i\leq l$ and
$j_i<j\leq k$ (this is because $\alpha$s of \ref{main thm} are calculated modulo $h_1$). As computed in the proof
of \ref{lem1}, $\dim_\F C^\bot-\dim_\F C'^\bot= \deg (g_{1,1})=m-\deg (h_{1,1})$, where $h_{1,1}$ is as in
\ref{main thm} (this is also true in the case that $\li(C)>1$ and $h_{1,1}=1$). Thus again by induction we see
that $\sum_{i=1}^k \deg(h_{i,j_i})= \dim_\F C^\bot$. Consequently, $H^R$ has properties \ref{div bas char 1},
\ref{div bas char 2} and \ref{div bas char 3'} of \ref{div bas char} and the result follows.
\end{proof}

Note that $C^\bot$ and $C^{\bot R}$ are equivalent codes. So by finding parameters and properties of one of these
codes, we have found those of the other one. We end this section with an example which applies Algorithm
\oldref{alg1} on the code $C$ in Example \ref{counterexam}.

\begin{ex}
Consider the code $C$ generated by $G$ in Example \ref{counterexam} over an arbitrary field.  If we run Algorithm
\oldref{alg1} on $G$, it returns
$$H=\l(\begin{array}{ccc}
x(x^3+1) & 0 & 0 \\
-(x^3+1) & x^3+1 & 0 \\
1 & -1 & x^2
\end{array}\r)\hbox{ and hence }
H^R=\l(\begin{array}{ccc}
x^2 & -1 & 1 \\
0 & x^3+1 & -(x^3+1) \\
0& 0& x(x^3+1)
\end{array}\r)$$
is the CGM of $C^{\bot R}$.
\end{ex}
%
%

                                  \section{Some Self-Dual $A$-Codes}
Self-dual codes have both theoretical and practical importance (see for example \cite{self dual}). A lot of effort
has been devoted to classify and enumerate self-dual codes of small or moderate length over different rings. For
example in \cite{SD ter 24}, all ternary self-dual codes of length 24 are classified and in \cite{SD froben}, it
is proved that over any finite commutative Frobenius ring, self-dual codes exist. Also in \cite{SD Zm}, several
classes of self-dual codes over rings $\z_m$ are classified, including self-dual codes of length 4 and 8 over
$\z_{pq}$, where $p$ and $q$ are distinct primes. The reader is referred to \cite{SD survey}, for a survey of
classification and enumeration of self-dual codes of small length over $\F_2,\, \F_3,\, \F_4,\, \z_4$ and
$\F_2+u\F_2$.

A source of importance for self-dual codes over finite fields is the MacWilliams identity (see for example
\cite[Section 5.4]{roman}). But $A=\F[x]/\lg f(x) \rg$ is a finite principal ideal ring (PIR) and hence a finite
Frobenius ring. So by \cite[Section 8]{wood} MacWilliams identity holds for $A$-codes.
Thus many of the important properties of self-dual codes over fields, hold for every $A$-code. Also note that
since the weight enumerator of $C^\bot$ is the same as the weight enumerator of $C^{\bot R}$, the codes which
equal their reciprocal dual are also of the same importance. In this section, we use \ref{main thm} to find some
self-dual $A$-codes. Our first result considers the case that bases of divisors of $C$ have just one element.
\begin{thm}\label{SD 1 base}
Suppose that $C$ has a basis of divisors consisting of one element, say $g=(g_1,\ldots, g_l)$. Then $C$ is
self-dual \ifof either $l=1$ and $f=g_1^2$ in $\F[x]$ or $l=2$, $g_1=1$ and $g_2^2=-1$ in $A$. Also $C=C^{\bot
R}$, \ifof one of the following hold:
\begin{enumerate}
\item $l=1$ and $f=g_1^2$ in $\F[x]$,

\item $l=2$, $g_1=0$ and $g_2=1$,

\item $l=2$, $g_1=1$ and $g_2=0$,

\item $l=2$, char $\F=2$, $g_1=1$ and $g_2$ is any polynomial in $A$.
\end{enumerate}
\end{thm}
\begin{proof}
First we consider the statement on self-duality. (\rgive): Straightforward. (\give): If $g_1=0$, then clearly $C$
is not self-dual. Also if $l=1$, then $C^\bot$ is generated by $h=f/g_1$ and hence $C$ is self-dual \ifof
$f/g_1=g_1$. Thus assume that $l>1$ and $g_1\neq 0$. It follows from property \ref{div bas char}\ref{div bas char
3} that $hg_i=0$ for all $i$. This means that $g_i=g'_ig_1$ for some $g'_i\in A$. So according to \ref{main thm},
a generator matrix for $C^\bot$ is
$$H= \l(\begin{array}{ccccc}
h & 0& 0 & \cdots &0 \cr
-g'_2 & 1 & 0 & \cdots & 0\cr
-g'_3 & 0 & 1 & \cdots& 0 \cr
\vdots &\vdots & \ddots & \ddots& \vdots \cr
-g'_l & 0 & 0 &\cdots & 1
\end{array}\r).$$
Therefore  for each $2\leq i\leq l$, the $i$-th row of $H =h^{(i)}\in C^\bot= C$, that is, $h^{(i)}=a_i g$ for
some $a_i\in A$. Hence $a_ig_1=-g'_i$ and $1=a_i g_i= a_i g'_i g_1=g'_i(-g'_i)$. Consequently, the $a_i$s and
$g'_i$s are units in $A$. As  $a_ig_1=-g'_i$, $g_1$ is also a unit of $A$. Since $g_1|f$, we conclude that
$g_1=1$. If $l>2$ we get $a_2 g_3=0$ which contradicts $a_2, g_2$ being units. Thus $l=2$ and since $g \cdot g=0$
we see that $g_2^2=-1$.

For $C=C^{\bot R}$, again (\rgive) and also (\give) for $l=1$ is easy. Suppose that $C=C^{\bot R}$ and $l>1$. If
$g_1=\cdots =g_i=0$ for some $0<i\leq l$, then according to \ref{rev canon} the CGM of $C^{\bot R}$ has at least
$i$ rows, one of which is $(0, 0, \ldots, 0, 1)$. Since this CGM should have just one row which is $g$, we deduce
that $i=1$, $l=2$, $g_1=0$ and $g_2=1$. Now assume that $g_1\neq 0$. By \ref{rev canon}, if $l>2$, then the CGM of
$C^{\bot R}$ has more than one row, a contradiction. Also the first row of the CGM of $C^{\bot R}$ is the
reciprocal of the last row of $H$ above. So $g_1=1$ and $-g_2=g_2$. This last condition always holds in
characteristic 2 and holds just for $g_2=0$ in other characteristics.
\end{proof}

A code $C$ is called \emph{isodual}, when $C^\bot$ is equivalent to $C$ by a permutation. The following result
together with the previous theorem characterize isodual $A$-codes with a basis of divisors consisting of one
element.
\begin{prop}
Suppose that $C$ has a basis of divisors consisting of one element. Then $C$ is isodual \ifof $l\leq 2$ and either
$C$ is self-dual or $C=C^{\bot R}$.
\end{prop}
\begin{proof}
Clearly we just need to show that if $C$ is isodual, then $l\leq 2$. Suppose $g=(g_1,\ldots, g_l)$ is the only
element of a basis of divisors of $C$. Assume that $g_1= \cdots= g_{i-1}=0$ and $g_i\neq 0$. Then it follows from
\ref{main thm} and Algorithm \oldref{alg1}, that a generator matrix for $C^\bot$ is
$$H= \mat{c|c}{
I_{(i-1)\times (i-1)} &  0 \cr
\hline
0 & \begin{array}{c|c}
     h & 0 \cr
     \hline
     * & I_{(l-i)\times (l-i)}
     \end{array}
},$$%
where $h=f/g_i \mod{f}$. Since $H^R$ is a CGM for $C^{\bot R}$ and by \cite[Proposition 2]{main}, we deduce that
$$m> \dim_\F C=\dim_\F C^\bot= \dim_\F C^{\bot R}> (l-1) m$$ and hence $l\leq 2$, as required.
\end{proof}
When the basis of divisors of an isodual $A$-code $C$ has more that one element, then we may have $l\geq 3$. For
instance, if $f(x)=x^2$ and $C$ has CGM $\mat{ccc}{0 & x & 0 \cr 0 & 0 & 1}$, then one can readily verify that
$C=C^{\bot R}$ and is an isodual code with length 3.

Next we present a characterization of self-dual $A$-codes of length = 2. Consider $g_1,g_2, g_3\in A$ and set
$h_i=f/g_i \mod{f}$ (if $g_i=0$, then $h_i=1$). Using \ref{div bas char}, we can distinguish three classes of
$A$-codes of length 2. Class I: those with CGM of the form $(g_1\ g_2)$ with $0\neq g_1|f$ and $h_1 g_2=0$
(equivalently, $g_1|g_2$). Class II: those with CGM $(0\ g_2)$ with $g_2|f$. Class III: those with CGM
 $\l(\begin{array}{cc}
 g_1 & g_2 \cr
 0 & g_3
 \end{array}\r)$,
where $0\neq g_1, g_3|f$, $g_3|h_1g_2$ and $\deg(g_2)<\deg(g_3)$. The generator matrix for the dual of these codes
calculated by Algorithm \oldref{alg1} is:
 $$\hbox{class I: }\mat{cc}{h_1 & 0 \cr -\f{g_2}{g_1} & 1},
 \hbox{ class II: }\mat{cc}{1 & 0 \cr 0 & h_2}
 \hbox{ and class III: }\mat{cc}{h_1 & 0 \cr -\f{h_3g_2}{g_1} & h_3}.$$
Therefore by \ref{rev canon}, we immediately get the following.
\begin{prop}\label{SRD}
An $A$-code $C$ of length 2 is equal to its reciprocal dual $C^{\bot R}$ \ifof either it is of class I and $g_1=1$
and when char $\F\neq 2$, $g_2=0$ or it is of class II with $g_2=1$ or it is of class III with $g_1g_3=f$ in
$\F[x]$ and when char $\F\neq 2$, $g_2=0$.
\end{prop}

The main part of characterizing self-dual $A$-codes of length 2 is the following.
\begin{thm}\label{SD l=2 II}
Suppose that $C$ is a class III code of length 2 with CGM
 $G=\l(\begin{array}{cc}
 g_1 & g_2 \cr
 0 & g_3
 \end{array}\r).$
Then the following are equivalent.
\begin{enumerate}
\item \label{SD l=2 II 1} C is self-dual.

\item \label{SD l=2 II 2} $\deg(g_1) + \deg(g_3)= m$, $0=g_3^2= g_2 g_3= g_1^2+g_2^2$ (in $A$) and $g_1^2|f$ (in
    $\F[x]$).

\item \label{SD l=2 II 3} There exist $g', f', r\in \F[x]$ with $g'^2=rf'-1$ such that $f=g_1^2 f'$, $g_3=g_1
    f'$ and $g_2=g_1 g'$.
\end{enumerate}
\end{thm}
\begin{proof}
\ref{SD l=2 II 1} \give \ref{SD l=2 II 2}: Since $\dim_\F C=\dim_\F C^\bot=2m-\dim_\F C$ and as $$\dim_\F
C=2m-\deg (g_1)-\deg (g_3)$$ by \cite[Proposition 2]{main}, it follows that $\deg(g_1) + \deg(g_3)= m$. Let
$h_i=f/g_i$. By assumption $(g_1,g_2)\in C=C^\bot$. Thus according to the above notes, $(g_1, g_2)=\alpha\p{
-\f{h_3g_2}{g_1}, h_3}+\beta (h_1, 0)$ for some $\alpha, \beta\in A$. Therefore, in $\F[x]$ we have $g_2=
\alpha\f{f}{g_3}+ kf\quad (*)$ for some $k\in \F[x]$. Hence $\f{f}{g_3}| g_2$, that is, $f|g_2g_3$. Similarly from
$(0, g_3)\in C^\bot$ we deduce that $f|g_3^2$. Also from $(*)$ we deduce that $\alpha=\f{g_2}{f}g_3-kg_3$ and
hence
$$g_1=-\alpha \f{h_3g_2}{g_1} +\beta h_1=-\p{ \f{g_2g_3}{f}\f{fg_2}{g_3g_1}} +kg_3 \f{fg_2}{g_3g_1} +\beta
\f{f}{g_1}= -\f{g_2^2}{g_1} +\f{kfg_2}{g_1} +\f{\beta f}{g_1}.$$ Multiplying by $g_1$ we get $g_1^2\equiv - g_2^2
\mod{f}$, as claimed.

On the other hand, $(h_1, 0)\in C^\bot=C$, and hence by \cite[Theorem 1]{main}, in $\F[x]$ we have $f/g_1=
h_1=\alpha' g_1$ for some $\alpha'$ with $\deg(\alpha')<\deg(h_1)$. Consequently this equality holds in $\F[x]$
and $g_1^2|f$.

\ref{SD l=2 II 2} \give \ref{SD l=2 II 3}: By assumption $g_1^2|f$. Set $f'=\f{f}{g_1^2}$. Also since
$rf=g_1^2+g_2^2$ for some $r\in\F[x]$ and $g_1^2|f$, we deduce that $g_1^2|g_2^2$, whence $g_1|g_2$. Let
$g'=\f{g_2}{g_1}$. Then $rf'=1+g'^2$. It remains to show $g_3=g_1f'$.

From the equation $rf'-g'^2=1$, we see that $(f',g')=1$. Also by assumption $f|g_3g_2=g_3g_1g'$, hence
$g_1f'=\f{f}{g_1}|g_3g'$. Since $(f', g')=1$, we deduce that $f'|g_3$ and $g_1|g'g'_3$ where $g'_3=\f{g_3}{f'}$.
But $g_1^2| f| g_3^2= {g'_3}^2{f'}^2$. So $g_1|f'g'_3$. Thus $g_1|(rf'-g'^2)g'_3=g'_3$. Now
$$\deg(g_3)+\deg(g_1)=m=\deg(f)=\deg(f')+2\deg(g_1),$$ whence $\deg(g_1)=\deg(g_3)-\deg(f')=\deg(g'_3)$. As both
$g_1$ and $g'_3$ are monic, we conclude that $g_1=g'_3$ and the result follows.

\ref{SD l=2 II 3} \give \ref{SD l=2 II 1}: Let $H=\mat{cc}{h_1 & 0 \cr -\f{h_3g_2}{g_1} & h_3}= \mat{cc}{g_1f' &
0\cr -g_1g' & g_1}$ which is a generator matrix for $C^\bot$ according to \ref{main thm}. One can readily check
that $G=\mat{cc}{r & g' \cr g' & f'} H$. So $C\se C^\bot$. But since $\deg(g_1)+\deg(g_3)=\deg(f)=m$, it follows
that $|C|=|C^\bot|$ and hence $C=C^\bot$.
\end{proof}

\begin{cor}\label{SD l=2 main}
Suppose that $C$ is an $A$-code of length 2. Then $C$ is self-dual \ifof its CGM is either $[1\ g_2]$ with
$g_2^2=-1$ or it is a class III code satisfying the equivalent conditions of \ref{SD l=2 II}.
\end{cor}

In particular, we get the following family of self-dual codes.
\begin{ex}\label{ex SD l=2}
Let $0,1 \neq g_1$ and $g'$ be any pair of monic polynomials in $\F[x]$. Then by \ref{SD l=2 II}, we see that the
$A$-code generated by the matrix
 $\l( \begin{array}{cc}
 g_1 & g_1 g' \cr
 0 & g_1(g'^2+1)
 \end{array}\r)$
is self-dual, where $A=\f{\F[x]}{\lg g_1^2 (g'^2+1)\rg }$. Also the $B$-code with CGM $\mat{cc}{g_1 & g_1x^3 \cr 0
& g_1(x^4-x^2+1)}$ is self-dual with $B=\f{\F[x]}{\lg g_1^2 (x^4-x^2+1) \rg}$.
\end{ex}

In the case that char $\F$=2 we can simplify the characterization of self-dual $A$-codes presented in \ref{SD l=2
II}. First we need a simple lemma.
\begin{lem}\label{^2 in F'=>in F}
Suppose that $\F\se \F'$ are finite fields and char $\F=p$. If $g\in \F'[x]$ is such that $g^p\in \F[x]$, then
$g\in \F[x]$.
\end{lem}
\begin{proof}
Suppose that $|\F'|=p^n$. Note that $g(x^{p^n})=(g(x))^{p^n}=((g(x))^p)^{p^{n-1}}\in \F[x]$. Thus all coefficients
of $g$ are in $\F$.
\end{proof}

In the sequel, by $\sqrt[2]{h}$ we mean $p_1^{\ceil{\alpha_1/2}} \cdots p_t^{\ceil{\alpha_t/2}}$, where
$h=p_1^{\alpha_1}\cdots p_t^{\alpha_t}$ is the prime decomposition of $h$ in $\F[x]$. Note that $g^2\in \lg h \rg$
\ifof $g\in \lg \sqrt[2]{h} \rg$.
\begin{thm}\label{SD in char 2}
Suppose that char $\F =2$ and $C$ is a class III code of length 2 with CGM
 $\l(\begin{array}{cc}
 g_1 & g_2 \cr
 0 & g_3
 \end{array}\r)$.
Then $C$ is self-dual \ifof
\begin{enumerate}
\item \label{SD in char 2_1} either $g_1=g_3$, $g_2=0$ and $f=g_1^2$,

\item \label{SD in char 2_2} or $f=g_1^2f'$, $g_3=g_1f'$ and $g_2=g_1(h\sqrt[2]{f'}+1)$ for some $f', h\in
    \F[x]$ with $\deg(h)< \deg(f')-\deg(\sqrt[2]{f'})$.
\end{enumerate}
\end{thm}
\begin{proof}
(\rgive): Follows from \ref{SD l=2 II}. (\give): Since $C$ is self-dual the conditions of \ref{SD l=2 II}\ref{SD
l=2 II 3} hold. Thus we just need to show that in the notations of \ref{SD l=2 II}\ref{div bas char 3}
$g'=h\sqrt[2]{f'}+1$ with $\deg(h)< \deg(f')-\deg(\sqrt[2]{f'})$. We have $g'^2=rf'-1=rf'+1$. If $f=1$, then
$g_1=g_3$ and $f_1=g_1^2$. Also since $G$ is a CGM, we have $\deg(g_2)< \deg(g_3)= \deg(g_1)$. But as $g_2$ is a
multiple of $g_1$ we get $g_2=0$ and case \ref{SD in char 2_1} is valid.

Thus assume that $f'\neq 1$. Therefore, $f'$ has some root, say $a$, in some extension field $\F'$ of $\F$.
Therefore in $\F'$ we have $g'^2(a)=r(a)f'(a)+1=1$, whence $g'(a)=1$. So $g'=g''(x-a) +1$ for some $g''\in \F'[x]$
and $rf'=g'^2-1=g''^2(x-a)^2$. Thus by \ref{^2 in F'=>in F}, $g''(x-a)\in \F[x]$. Since $g''^2(x-a)^2\in \lg f'
\rg$, we deduce that $g''(x-a)=h\sqrt[2]{f'}$ for some $h\in \F[x]$, that is, $g'=h\sqrt[2]{f'}+1$. Noting that
$\deg(g_2)<\deg(g_3)$ and $g_2=g'g_1$ and $g_3=f'g_1$, the degree condition on $h$ follows and the proof is
complete.
\end{proof}

At the end of this section, we pay some attention to some self-dual codes of length $>2$. Suppose that $C_1$ and
$C_2$ are two $A$-codes of length $l_1$ and $l_2$, respectively, and denote their direct product by $C= C_1\times
C_2$. Then according to \cite[Lemma 3.2]{SD froben}, $C$ is a self-dual $A$-code of length $l_1+l_2$. Thus if a
self-dual $A$-code of length 2 exists, then as in the following example we can construct a self-dual $A$-code of
any given even length. It is straightforward to check that if $G_1$ and $G_2$ are CGMs of $C_1$ and $C_2$,
respectively, then the block matrix $\mat{c|c}{
 G_1 & 0\\
 \hline
 0 & G_2}$
is the CGM of $C_1\times C_2$.
\begin{ex}
In the notations of Example \ref{ex SD l=2} and according to the above notes, the following is the CGM of a
self-dual $B$-code of length 4:
 $$\mat{cccc}{
 g_1 & g_1x^3 & 0&0\\
 0 & g_1(x^4-x^2+1) &0 &0 \\
 0 & 0& g_1 & g_1x^3\\
 0& 0& 0 & g_1(x^4-x^2+1)}.$$
\end{ex}
Let $R$ be any finite ring and $\fm$ one of its maximal ideals. Then the \emph{stability index }of $\fm$ is the
least positive integer such that $\fm^t=\fm^{t+1}=\cdots$. Using \cite[Theorem 3.9]{SD froben}, we can deduce the
following result on lengths of self-dual $A$-codes. For its proof, we use the well-known fact that $-1$ is a
square in a field $\F$ \ifof either $|\F|$ is even or $|\F|\equiv 3 \mod{4}$ (it follows, for example, from
\cite[Excercise 2.13]{lidl}).
\begin{thm}
If $l$ is a multiple of 4, then there exist a self-dual $A$-code of length $l$. Also there exist self-dual
$A$-codes of all lengths \ifof $f$ is a square. Moreover, suppose $f=f_1^{e_1}\cdots f_t^{e_t}$ is the prime
decomposition of $f$ in $\F[x]$ and $d_i=\deg(f_i)$. Then the following are equivalent.
\begin{enumerate}
\item \label{1} There exist self-dual $A$-codes of all even lengths.
\item \label{2} $f=g^2f'$ for some $g,f'\in \F[x]$ such that $-1$ is a square modulo $f'$ (that is, a square in
    $\F[x]/\lg f'(x) \rg$).
\item \label{3} Either $|F|\equiv 3\mod{4}$ and for each $1\leq i\leq t$, $d_ie_i$ is even or $|F|$ is even or
    $|F|\equiv 1 \mod{4}$.
\end{enumerate}
\end{thm}
\begin{proof}
Every finite PIR, including $A$, satisfies the conditions of \cite[Theorem 3.9]{SD froben}. Therefore by that
theorem, there exist a self-dual $A$-code with length $l$, if $4|l$. Also according to the construction mentioned
in the previous example and the notes before it, we see that there are self-dual codes of every length \ifof there
is a self-dual code of length 1 which is clearly equivalent to $f$ being a square. To prove the last part of the
statement, let $\fm_i$ be the maximal ideal $\lg f_i\rg$ of $A$ and $\F_i=A/\fm_i\cong \F[x]/\lg f_i \rg$. Then
the stability index of $\fm_i$ is $e_i$ and $|\F_i|=|\F|^{d_i}$.

\ref{1} \give\ \ref{2}:  There exist a self-dual $A$-code $C$ of length 2. Now the result follows from \ref{SD l=2
main} and \ref{SD l=2 II}, if we set $g=1$ and $f'=f$ in the case that $C$ is a class II code and $g=g_1$ and $f'$
as in \ref{SD l=2 II 3} in the case that $C$ is a class III code.

\ref{2} \give\ \ref{3}: Suppose that \ref{2} holds but \ref{3} does not. Then $|\F| \equiv 3\mod{4}$ and there is
an $i$ such that both $d_i$ and $e_i$ are odd. Hence $f_i| f'$ and since $-1$ is a square modulo $f'$, it is also
a square modulo $f_i$. This means that $-1$ is a square in $\F_i$ and $|\F_i|$ is either even or $|\F_i|\equiv 1
\mod{4}$. But as $d_i$ is odd, we deduce that $|\F_i|=|\F|^{d_i}\equiv 3 \mod{4}$ and from this contradiction the
result follows.

\ref{3} \give\ \ref{1}: This follows \cite[Theorem 3.9(ii)]{SD froben}.
\end{proof}

       \section{Rings over Which the $\F$-Dual of Every Linear Code Is a Linear Code}
A polynomial in $A$ can be viewed as the vector of its coefficients in $\F$. Similarly a codeword $(g_1(x),
\ldots, g_l(x))$ can be viewed as the vector of length $lm$ over $\F$ obtained by concatenating the vectors
corresponding to $g_1(x)$, \ldots, $g_l(x)$. In this  way, every $A$-code of length $l$ is also a linear code of
length $lm$ over $\F$ and its $F$-dual can be computed. As Example 7 of \cite{main} shows, the $F$-dual of an
$A$-code need not be an $A$-code. In this section, we characterize monic polynomials $f(x)\in \F[x]$ with the
property that the $F$-dual of every $A$-code is an $A$-code, where $A=\f{\F[x]}{\lg f(x) \rg}$. For simplicity,
throughout this section we fix the following notations.

\begin{nota}
Let $g(x)=\sum_{i=0}^{m-1} a_ix^i\in A$. We can regard $g$ as the row vector $(a_0, \ldots, a_{m-1})$ over $\F$.
We denote this vector by $g$ or $[g(x)]$ and whenever we want to consider $g$ as a polynomial, we write $g(x)$
(not $g$ alone). Similarly if $u=(u_0, \ldots, u_{m-1})$ is a vector over $\F$, then by $u(x)$ we mean
$\sum_{i=0}^{m-1} u_ix^i$. Also, as in \cite{main}, we set
 $$M_x=\mat{cccccc}{
 0 & 1 &0 &\ldots &  \ldots & 0\cr
 \vdots &0 & 1 & 0 & \ddots & \vdots \cr
 \vdots & \vdots & \ddots & \ddots & \ddots & \vdots \cr
 0 & 0 & \ddots & \ddots &1 &0 \cr
 0 & 0 & \ldots & \ldots & 0 & 1 \cr
 -f_0 & -f_1 & \ldots & \ldots & \ldots & -f_{m-1}
 },$$
to be the companion matrix of $f(x)=x^m+\sum_{i=0}^{m-1} f_ix^i$. Moreover, for arbitrary $g(x)=\sum_{i=0}^{m-1}
a_ix^i \in A$, we set $M_g=g(M_x)=\sum_{i=0}^{m-1} a_i M_x^i$. Furthermore, we write $C^\bot$ for the $A$-dual of
$C$ and $C^{\bot_\F}$ for the $\F$-dual of $C$. \qed
\end{nota}
Consequently, it follows that $g(x)h(x)=(gM_h)(x)$ for any $g(x),h(x)\in A$, (see \cite[Proposition 6]{main}). To
find out when the $\F$-dual of $A$-codes are $A$-codes, we need the following.

\begin{prop}\label{M_x^T}
Assume that $m\geq 2$. There exists $g(x)\in A$ with $M_x^T=M_g$ \ifof either $f(x)=x^m\pm 1$ or $m=2$ and
$f(x)=x^2+ax-1$ for some $a\in \F$.
\end{prop}
\begin{proof}
(\give): It is easy to check that for each $1\leq i \leq m-1$, the first row of $M_x^i$ is $e_i$ (the vector with
just one nonzero entry which is a 1 on the $i$-th place) and the $m-i+1$-th row of $M_x^i$ is $(-f_0, -f_1,
\ldots, -f_{m-1})$. So if $g(x)=\sum_{i=0}^{m-1}a_i x^i$, then the first row of $M_g=\sum_{i=0}^{m-1} a_i M_x^i$
is $(a_0,a_1, \ldots, a_{m-1})$. On the other hand the first row of $M_x^T$ is $-f_0e_m$. Thus $a_i=0$ for all
$0\leq i\leq m-2$ and $a_{m-1}=-f_0$, that is, $g(x)=-f_0 x^{m-1}$.

Now it follows from the above notes that the second row of $M_g$ is $-f_0(-f_0, \ldots, -f_{m-1})$, which should
be equal to the second row of $M_x^T=(1, 0, 0, \ldots, 0,- f_1)$. Hence $f_0^2=1$ and for each $1\leq i\leq m-2$
we have $f_i=0$ and $f_0f_{m-1} =-f_1$. This, if $m>2$, results to $f_{m-1}=0$ (that is, $f(x)=x^m\pm 1$) and if
$m=2$, results to either $f_1=0$ (that is, $f(x)=x^2\pm 1$) or $f_0=-1$ (that is, $f(x)=x^2+ax-1$ for some $a\in
\F$).

(\rgive): It is routine to verify that in all cases $M_x^T=M_g$, where $g(x)=-f_0x^{m-1}$.
\end{proof}

Now we can state the main theorem of this section.
\begin{thm}\label{FD is A-code}
The $\F$-dual of every $A$-code is an $A$-code \ifof either $m=1$ or $m=2$ and $f(x)=x^2+ax-1$ for some $a\in \F$
or $m\geq 2$ and $f(x)=x^m\pm 1$.
\end{thm}
\begin{proof}
The case $m=1$ is trivial, so we assume that $m\geq 2$. Here if $C$ is a code of length $l$, $u\in C$ and $M$ is a
$m\times m$ $\F$-matrix, we regard $u$ as the vector $(u_1,\ldots, u_l)$ with $u_i$s in $A$ and write $uM$ for the
vector $(u_1M, \ldots, u_l M)$.

(\give): Suppose $C$ is an $A$-code and $z\in CM_x^T$. Since $C^{\bot_\F}$ is an $A$-code, we have $xu(x) \in
C^{\bot_\F}$, for each $u\in C^{\bot_\F}$. Therefore
 $$uz^T \in u\p{ CM_x^T} ^T= uM_x C^T =[xu(x)]C^T\se C^{\bot_\F}C^T=0.$$
Consequently, $CM_x^T\se C^{\bot_\F {\bot_\F}}=C$ for each $A$-code $C$. In particular, if $0\neq v\in A$, then
$vM_x^T\in Av$, where $Av$ is the ideal (or equivalently, the $A$-code of length 1) generated by $v$. This means
that $vM_x^T=[v(x)g_v(x)]$ for some $g_v(x)\in A$. If $0,v\neq v'\in A$, then $(v,v')M_x^T\in A(v,v')$. Hence for
some $g(x)\in A$, we have
 $$([g(x)v(x)], [g(x)v'(x)])= (v,v')M_x^T= (vM_x^T, v'M_x^T)= ([g_v(x)v(x)], [g_{v'}(x)v'(x)]).$$
If we apply this to $v'(x)=1$, we see that $g(x)=g_1(x)$. Therefore, for arbitrary $v(x)\in A$, we have $vM_{g_1}=
[g_1(x)v(x)]= vM_x^T$, that is $M_x^T= M_{g_1}$ and the result follows from \ref{M_x^T}.

(\rgive): By \ref{M_x^T}, there is a $g\in A$ with $M_x^T=M_g$. Since $C$ is an $A$-code, we see that $CM_g\se C$.
Therefore, if $u\in C^{\bot_\F}$, then
$$[xu(x)]C^T=uM_xC^T=u(CM_x^T)^T=u(CM_g)^T\se uC^T=0,$$ that is $xu(x)\in C^{\bot_\F}$, hence $C^{\bot_\F}$ is an $A$-code.
\end{proof}

Next we are going to find a generator matrix for $C^{\bot_\F}$ over $A$, where $A$ satisfies the conditions of
\ref{FD is A-code}. For this we need some intermediate results. Recall that if $g\in \F[x]$, then $g^R(x)$ is the
\emph{reciprocal} of $g(x)$, that is, $x^{\deg(g)} g(x^{-1})$.
\begin{lem}\label{M_g^T}
Assume that $g(x)\in A$. If $f(x)=x^2+ax-1$ for some $a\in \F$, then $M_g^T=M_g$. If $f(x)=x^m\pm 1$, then
$M_g^T=M_h$, where $h(x)= g(x^{-1})= \f{g^R(x)}{x^{\deg(g)}}$.
\end{lem}
\begin{proof}
In either of the cases, if $g(x)=\sum_{i=0}^{m-1} a_i x^i$, then (by the proof of \ref{M_x^T}) $$M_g^T=
\sum_{i=0}^{m-1} (M_x^T)^i= \sum_{i=0}^{m-1}(-f_0M_x^{m-1})^i= M_{h'},$$ where $h'(x)=g(-f_0x^{m-1})$. Now if
$f(x)=x^2+ax-1$, then $-f_0x^{m-1}=x$ and if $f(x)=x^m\pm 1$, then $-f_0x^{m-1}=x^{-1}$, thus $h'(x)$ is same as
$h(x)$ of the statement.
\end{proof}

Assume that $G=(g_{ij}(x))$ is a $k\times l$ matrix over $A$. As in \cite[Section 3.5]{main}, we set $\psi(G)$ and
$\zeta(G)$ to be the $km\times lm$ matrices over $\F$ defined blockwise as follows: the $ij$-th block of $\psi(G)$
is $M_{g_{ij}}$ and the $ij$-th block of $\zeta(G)$ is $M_{g_{ij}}^T$. According to \cite[Theorem 4]{main}, the
code that $\psi(G)$ generates over $\F$ is the same as the code $G$ generates over $A$ and by \cite[Theorem
5]{main}, $\zeta(G)$ generates $(C^\bot)^{\bot_\F}$ over $\F$.
\begin{cor}\label{Fdual=Adual}
For every $A$-code $C$ we have $C^\bot=C^{\bot_\F}$ \ifof $m=1$ or $f(x)=x^2+ax-1$ for some $a\in \F$.
\end{cor}
\begin{proof}
(\rgive): The case $m=1$ is trivial, thus assume that $f(x)=x^2+ax-1$. Let $G$ be a generator matrix for an
arbitrary $A$-code $C$. Then by \ref{M_g^T}, we have $\psi(G)=\zeta(G)$. Therefore, according to Theorems 4 and 5
of \cite{main}, $(C^\bot)^{\bot_\F}=C$ for every $A$-code $C$. Applying this with $C^\bot$ instead of $C$, we get
the desired conclusion.

(\give): Suppose $m\geq 2$ and $f(x)\neq x^2+ax-1$ for any $a\in \F$. Then by \ref{FD is A-code}, we should have
$f(x)=x^m\pm 1$. Assume that $f(x)=x^m-1$. Then $m>2$. Consider the code $C$ generated by $(1, x+1)$ over $A$ or
equivalently generated by the vectors $$([1], [x+1]), ([x], [x^2+x]), \ldots, ([x^{m-1}], [x^{m-1}(x+1)])$$ over
$\F$. Then it is routine to check that the dot product of the vector $([x^{m-1}+1],[-1])$ with any of the above
$\F$-generators of $C$ is zero, that is, $(x^{m-1}+1, -1)\in C^{\bot_\F}$ but $$(x^{m-1}+1, -1)(1,
x+1)=x^{m-1}-x\neq 0,$$ that is, $(x^{m-1}+1, -1)\notin C^{\bot}$, a contradiction.

Now assume that $f(x)=x^m+1$. If $m=2$ and char $\F=2$, then $f(x)$ is in the required form. Thus we can assume
that either $m>2$ or char $\F\neq 2$. Again consider the code $C$ generated by $(1, x+1)$ over $A$. Using these
assumptions one can readily verify that this time $(x^{m-1}-1, 1)\in C^{\bot_\F}\sm C^{\bot}$.
\end{proof}

\begin{prop}\label{zeta(G)}
Suppose that $f(x)=x^m\pm 1$ and assume that rows of a $k\times l$ matrix $G$ form a basis of divisors for an
$A$-code $C$. Let $G'= (\alpha_i^{-1} x^{d_i} g_{ij}(x^{-1}))$, where $d_i=\deg(\lc(g^{(i)}))$ and $\alpha_i$ is
the constant coefficient of $\lc(g^{(i)})$. Then rows of $G'$ form a basis of divisors for the $A$-code
$C'=(C^{\bot})^{\bot_\F}$.
\end{prop}
\begin{proof}
Let $G''=(g_{ij}(x^{-1}))$. Then as $x$ is invertible in $A$, $G'$ and $G''$ generate the same code. Now
$\psi(G'')=\zeta(G)$ by \ref{M_g^T}. Thus by \cite[Theorem 5]{main} $\psi(G'')$ generates $C'$ over $\F$ and hence
by \cite[Theorem 4]{main} $G''$ and hence $G$ generate $C'$ over $A$. Let $h_i(x)=\lc(g^{(i)})$, then
$h_i(x)|f(x)$. Now since $\lc\p{ g'^{(i)}}= \alpha_i^{-1} h_i^R(x)$ and $f^R(x)=\pm f(x)$, we see that $\lc\p{
g'^{(i)}} |f(x)$. Thus \ref{div bas char}\ref{div bas char 2} holds and obviously \ref{div bas char}\ref{div bas
char 1} also holds. Moreover, since $h_i(x)|f(x)$, we see that $h_i(0)\neq 0$ and thus
$\deg(h_i(x))=\deg(h_i^R(x))$. Therefore, $$\sum_{i=1}^k m-\deg\p{ \lc\p{ g'^{(i)}} } = \sum_{i=1}^k m-\deg\p{
\lc\p{ g^{(i)}} } =\dim_\F C= \dim_\F C'$$ and \ref{div bas char 3'} of \ref{div bas char} holds and the result
follows.
\end{proof}

The matrix $G'$ constructed above need not be a CGM, even if the initial $G$ is a CGM for $C$, as the following
example shows.
\begin{ex}\label{not can ex}
Let $f(x)=x^3-1$ and $C$ be the $A$-code with CGM $\mat{cc}{x^2+x+1 & -1 \cr 0 & x-1}$. Then $G'=\mat{cc}{x^2+x+1
& -x^2 \cr 0 & x-1}$ which is not a CGM. Indeed the CGM for $C'$ is $G$ itself and $C=C'$ in this case.
\end{ex}

We say that rows of a matrix $G$ is a \emph{reverse basis of divisors} for an $A$-code $C$, when the rows of the
matrix obtained by reversing the order of both rows and columns of $G$ (as in \ref{rev canon}) are a basis of
divisors for $C^R$. For example, rows of $H$ in \ref{rev canon} form a reverse basis of divisors for $C^\bot$.
\begin{cor}
Assume that $A$ satisfies the conditions of \ref{FD is A-code}. Suppose that $H=(h_{ij}(x))$ is the matrix
obtained by Algorithm \oldref{alg1} for the $A$-code $C$. If $m=1$ or $f(x)=x^2+ax-1$, then rows of $H$ form a
reverse basis of divisors for $C^{\bot_\F}$. If $f(x)=x^m\pm 1$, then rows of $H'=(\alpha_i^{-1} x^{d_i}
h_{ij}(x))$ form a reverse basis of divisors for $C^{\bot_\F}$, where $d_i$ is the degree of the last nonzero
entry on the $i$-th row of $H$ and $\alpha_i$ is the constant coefficient of this entry.
\end{cor}
\begin{proof}
In the first case clearly $C^{\bot_\F}=C^\bot$ by \ref{Fdual=Adual}, so assume that $f(x)=x^m\pm 1$. Let $H^R$ be
as in \ref{rev canon}. If we apply \ref{zeta(G)} with $G=H^R$, then $G'=H'^R$ and $C'=((C^{\bot R})^\bot)
^{\bot_\F}=(C^R)^{\bot_\F}= (C^{\bot_\F})^R$ (note that in all terms, we are taking reciprocal of codes as
$A$-codes not $\F$-codes). Therefore, the rows of $H'^R$ form a basis of divisors for $(C^{\bot_\F})^R$, as
required.
\end{proof}
Note that although $H^R$ above is indeed a CGM for $C^{\bot R}$, but $H'^R$ need not be a CGM. For example, if
$f(x)=x^3-1$ and $C$ is the code with CGM $\mat{cc}{x^2+x+1 & 1\cr 0& x-1}$, then $H^R$ and $H'^R$ are $G$ and
$G'$ in Example \ref{not can ex}, respectively. Also if we apply this corollary for example to the code generated
by $(1, x+1)$ with $f(x)=x^m+1$, we see that $C^{\bot_\F}$ is generated by $(x^{m-1}-1, 1)$ and is different from
$C^\bot$ which is generated by $(-x-1,1)$.

{\large\bf Acknowledgement} \\
The author would like to thank Prof. H. Sharif of Shiraz University for his nice comments and discussions. Also
the author thanks the reviewers for careful reading of the paper and their valuable comments.


\end{document}